\documentclass[conference]{IEEEtran}

\usepackage{cite}
\usepackage{algorithm,
algorithmic, float}
\usepackage{caption}
\usepackage[dvips]{graphicx}
\usepackage{multirow}
\usepackage{colortbl}
\usepackage{tabularx}
\usepackage{subfigure}
\usepackage{color}
\usepackage{graphicx}
\usepackage{setspace}
\usepackage{mathtools}
\usepackage{bm}
\usepackage{amsmath,amsfonts,amssymb,amsthm,epsfig,epstopdf,url,array}
\usepackage{dsfont}

\DeclarePairedDelimiter{\floor}{\lfloor}{\rfloor}

\begin{document}
%\setcounter{page}{1}
%\pagenumbering{arabic}
\title{
\onehalfspacing
Capacity of Clustered Distributed Storage
}

\author{\IEEEauthorblockN{Jy-yong Sohn, Beongjun Choi, Sung Whan Yoon, and Jaekyun Moon}
\IEEEauthorblockA{School of Electrical Engineering\\
Korea Advanced Institute of Science and Technology\\
Daejeon, 34141, Republic of Korea\\
Email: \{jysohn1108, bbzang10, shyoon8\}@kaist.ac.kr, jmoon@kaist.edu}
}



%Telephone: (800) 555--1212, Fax: (888) 555--1212}


\maketitle

\begin{abstract}
%\boldmath
A new system model reflecting the clustered structure of distributed storage is suggested to investigate bandwidth requirements for 
repairing failed storage nodes. Large data centers with multiple racks/disks or local networks of storage devices (e.g. sensor network) are good applications of the suggested clustered model.
In realistic scenarios involving clustered storage structures, repairing storage nodes using intact nodes residing in other clusters is more bandwidth-consuming than restoring nodes based on information from intra-cluster nodes. Therefore, it is important to differentiate between intra-cluster repair bandwidth and cross-cluster repair bandwidth in modeling distributed storage.  
Capacity of the suggested model is obtained as a function of fundamental resources of distributed storage systems, namely, storage capacity, intra-cluster repair bandwidth and cross-cluster repair bandwidth. Based on the capacity expression, feasible sets of required resources which enable reliable storage are analyzed. It is shown that the cross-cluster traffic can be minimized to zero (i.e., intra-cluster local repair becomes possible) by allowing extra resources on storage capacity and intra-cluster repair bandwidth, according to a law specified in a closed-form. Moreover, trade-off between cross-cluster traffic and intra-cluster traffic is observed for sufficiently large storage capacity. 

%This document summarizes the mathematical result \& focusing scenario of J. Sohn's work on Distributed Storage Systems (DSS). B-J Choi and S.W. Yoon contributed to this work. Specifically, S.W. Yoon introduced overall concept and expected benefit of clustered DSS, while B-J Choi helped the mathematical analysis on the capacity.

%This work deals with capacity of clustered DSS and fundamental lower bound of repair bandwidth communicated across clusters (which we call cross-cluster repair BW). Dimakis' previous work \cite{dimakis2010network} on DSS investigated fundamental trade-off between storage size and repair bandwidth (BW), but did not consider the clustered structure of DSS. In this work, new system model considering clustered DSS is suggested, with min-cut analysis to obtain fundamental lower bound of cross-cluster repair BW.  
\end{abstract}


% \fi
\IEEEpeerreviewmaketitle

\section{Introduction}

Many enterprises use cloud storage systems in order to support massive amounts of data storage requests from clients.
In the emerging Internet-of-Thing (IoT) era, the number of devices which generate data and connect to the network increases exponentially, so that efficient management of data center becomes a formidable challenge.
However, since a cloud storage system consists of inexpensive commodity disks, failure events occur frequently, degrading system reliability \cite{ghemawat2003google}.

In order to ensure reliability of cloud storage, distributed storage systems (DSSs) with erasure coding have been considered to improve tolerance against storage node failures \cite{bhagwan2004total,dabek2004designing,rhea2003pond}. %\textcolor{red}{Examples of DSS storage in enterprises goes here.} 
In such systems, the original file is encoded and distributed into multiple storage nodes. When a node fails, a newcomer node regenerates the failed node by contacting a number of survived nodes. This causes traffic burden across the network taking up significant repair bandwidth.
The pioneering work of \cite{dimakis2010network} on distributed storage found the system capacity $\mathcal{C}$, the maximum amount of reliably storable data given the storage size of each node as well as repair bandwidth. A fundamental trade-off between storage node size and repair bandwidth was obtained, which ensures a reliable storage system with the maximum-distance-separable (MDS) property (i.e., any $k$ out of $n$ storage nodes can be accessed to recover the original file). The authors related the failure-repair process of a DSS with the multi-casting problem in network information theory, and exploited the fact that cut-set bound is achievable by network coding \cite{ahlswede2000network}. 

This fundamental result is based on the assumption of a homogeneous system, i.e., each node has the same storage size and repair bandwidth. 
However, storage nodes are actually dispersed into multiple clusters (called disks or racks) in data centers \cite{ford2010availability,huang2012erasure,muralidhar2014f4}, allowing high reliability against both node and rack failures. 
In this clustered system, repairing a failed node gives rise to both intra-cluster and cross-cluster repair traffic.
Since the available cross-rack communication bandwidth is typically $5-20$ times lower than the  intra-rack bandwidth in practical systems \cite{ahmad2014shufflewatcher}, a new system model which reflects this imbalance is required.

\hspace*{6pt} \textit{Main Contributions:} This paper suggests a new system model for \textit{clustered DSS} to reflect the clustered nature of real distributed storage systems wherein an imbalance exists between intra and cross-cluster repair burdens.
This model can be applied to not only large data centers, but also local networks of storage devices 
such as the sensor networks or home clouds which are expected to be prevalent in the IoT era.
This model is also more general in the sense that when the intra- and cross-cluster repair bandwidths are set equal, the resulting structure reduces to the original DSS model of \cite{dimakis2010network}. 

Storage capacity of the clustered DSS is obtained as a function of node storage size, intra-cluster repair bandwidth and cross-cluster repair bandwidth. 
The existence of the cluster structure manifested as the imbalance between intra/cross-cluster traffics makes the capacity analysis challenging;
Dimakis' proof in \cite{dimakis2010network} cannot be directly extended to handle the problem at hand.
We show that symmetric repair (obtaining the same amount of information from each helper node) is optimal in the sense of maximizing capacity given the storage node size and total repair bandwidth, as also shown in \cite{ernvall2013capacity} for the case of varying repair bandwidth across the nodes. 
However, we stress that in most practical scenarios, the need is greater for reducing cross-cluster communication burden, and we show that this is possible by trading with reduced overall storage capacity and/or increasing intra-repair bandwidth.

The behavior of the clustered DSS is also observed based on the derived capacity expression. 
First, the condition for \textit{zero} cross-cluster repair bandwidth is obtained, which enables local repair via intra-cluster communication only.
In order to support local repair while preserving the MDS property, extra amounts of resources are required, according to a 
precise mathematical rule derived in a closed-form.
Secondly, trade-off between cross-cluster repair bandwidth and intra-cluster repair bandwidth is discussed. The repair bandwidth resource is required to ensure reliability of a file against successive failure events; the amounts of intra-/cross-cluster repair bandwidth can be chosen in the trade-off curve, depending on the system constraint.

\hspace*{6pt} \textit{Related works:}
In order to reflect the non-homogeneous structure of storage nodes, several researchers in the literature aimed to analyze practical distributed storage systems. A heterogeneous model was considered in \cite{ernvall2013capacity}, where storage size of each node and repair bandwidth for each newcomer node is generally non-uniform. Upper/lower capacity bounds for the heterogeneous DSS are obtained in \cite{ernvall2013capacity}. 
A asymmetric repair process is considered in \cite{akhlaghi2010fundamental}, coining cheap and expensive nodes, depending on the amount of data that can be transfered to any newcomer. This is different from our analysis where we adopt a notion of cluster and introduce imbalance between intra- and /cross-cluster repair burdens.

In \cite{gaston2013realistic}, the idea of \cite{akhlaghi2010fundamental} is developed to a two-rack system, by setting the communication burden within a rack much lower than the burden across different racks, similar to our analysis. However, they classified the number of helper nodes by each rack. 
In other words, any failed node is repaired by $d_1$ helper nodes in $1^{st}$ rack and $d_2$ helper nodes in $2^{nd}$ rack, irrespective of the location of the failed node. On the other hand, the present paper classifies the number of helper nodes by their locations relative to the failed node. A failed node is repaired by $d_I$ helper nodes within the same rack and $d_c$ helper nodes in other racks. Compared to \cite{gaston2013realistic}, the setting in our work allows insights into tradeoffs between the intra- and cross-rack repair bandwidth.
Moreover, \cite{gaston2013realistic} focused on a repair cost function based on a weighted sum of two types of repair bandwidths.
whereas this paper investigates the trade-off between the two different repair bandwidth types as well as the storage node size.
Finally, some other works focused on the code design applicable to multi-rack DSSs \cite{tebbi2014code, 7541298}.

\hspace*{10pt} \textit{Organization:} This paper is organized as follows. Section II describes basic preliminary materials about distributed storage systems and the information flow graph, an efficient tool for analyzing DSS. A new system model for the clustered DSS is suggested in Section III, while the capacity of the suggested model is obtained in Section IV. Based on the derived capacity expression, the key nature of the clustered DSS is unveiled in Section V. Finally, Section VI draws the conclusion.

\section{Background}

\subsection{Distributed Storage System}
Distributed storage systems have been considered as a candidate for storing data, which maintains reliability by means of erasure coding \cite{dimakis2011survey}. The original data file is spread into $n$ unreliable nodes, each with storage size $\alpha$. When a node fails, it is regenerated by contacting $d < n$ helper nodes and obtaining a particular amount of data, $\beta$, from each helper node. The amount of communication burden allocated for one failure event is called the repair bandwidth, denoted as $\gamma = d \beta$.
When the client requests a retrieval of the original file, assuming all failed nodes have been repaired, access to any $k<n$ out of $n$ nodes must guarantee a file recovery. The ability to recover the original data using any any $k<n$ out of $n$ nodes is based on the MDS property.   
Distributed storage systems can be used in many applications such as large data centers, peer-to-peer storage systems and wireless sensor networks \cite{dimakis2010network}. 

\subsection{Information Flow Graph}
Information flow graph is a useful tool to analyze the amount of information flow from source to data collector in a DSS. 
It is a directed graph consisting of three types of nodes: data source $\mathrm{S}$, data collector $\mathrm{DC}$, and storage nodes $x_{in}^i, x_{out}^i$ as shown in Fig. \ref{Fig:information_flow_graph}.
Storage node $x^{i}$ can be viewed as consisting of 'input-node' $x_{in}^i$ and 'output node' $x_{out}^i$, which are responsible for the incoming and outgoing edges, respectively. $x_{in}^i$ and $x_{out}^i$ are connected by a directed edge with capacity identical to the storage size $\alpha$ of node $x^i$. 

Data from the source is stored into $n$ nodes. This process is represented by $n$ edges going from $\mathrm{S}$ to $\{x^i\}_{i=1}^n$, where each edge capacity is set to infinity.
A failure/repair process in a DSS can be described as follows.
When a node $x^{j}$ fails, a new node $x^{n+1}$ joins the graph by connecting edges from $d$ survived nodes, where each edge has capacity $\beta$.
After all repairs are done, data collector $\mathrm{DC}$ chooses arbitrary $k$ nodes to retrieve data, as illustrated by the edges connected 
from $k$ survived nodes with infinite edge capacity. 
Fig. \ref{Fig:information_flow_graph} gives an example of information flow graph representing a distributed storage system with $n=4, k=3, d=3$.

For a given flow graph $G$, a cut between $\mathrm{S}$ and $\mathrm{DC}$ is defined as a subset $C$ of edges which satisfies the following: every directed path from $\mathrm{S}$ to $\mathrm{DC}$ includes at least one edge in $C$. The cut between $\mathrm{S}$ and $\mathrm{DC}$ having the smallest total sum of edge capacities is called minimum cut.
A min-cut value is defined as the sum of edge capacities included in the minimum cut.

  \begin{figure}[!t]
	\centering
	\includegraphics[height=35mm]{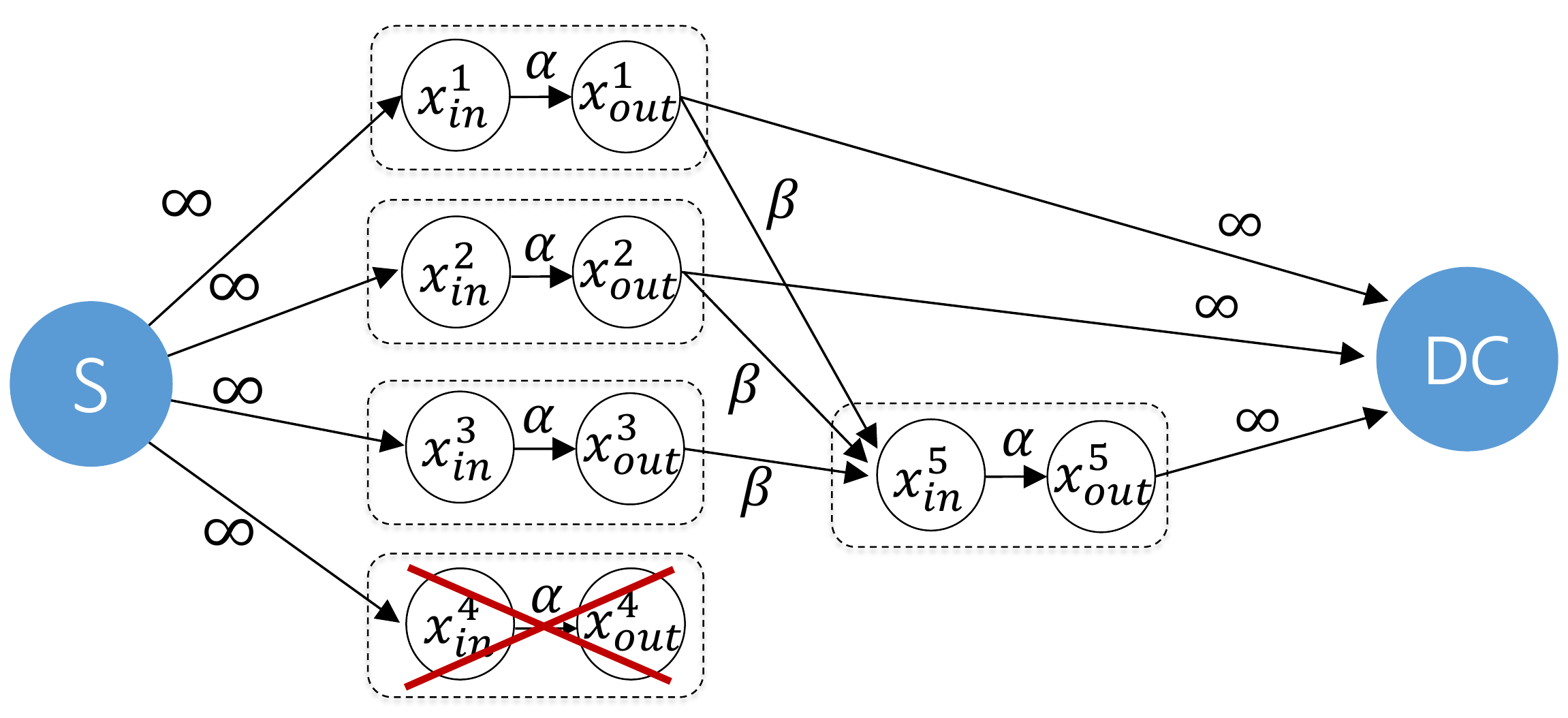}
	\caption{Information flow graph ($n=4, k=3, d=3$)}
	\label{Fig:information_flow_graph}
\end{figure}

%\subsection{Fundamental trade-off in non-clustered DSS}
%
%According to the max-flow min-cut theorem \textcolor{red}{Put reference here}, the amount of maximum transferable data from source to destination is equal to $c_{min}$($G$). Depending on the (possibly infinite) failure/repair process and $k$ nodes contacted from $DC$, various flow graph can be obtained. 
%Consider the set $\mathcal{G}$ of all possible flow graphs, and the min-cut minimizer $G^* \in \mathcal{G}$, i.e., $c_{min}$($G^*$) $\leq$ $c_{min}$($G$) for any $G \in \mathcal{G}$. 
%Then, any evolution of distributed storage system are guaranteed to transfer $c_{min}(G^*)$ amount of data from source to data collector, where the guaranteed amount is defined as capacity \cite{dimakis2010network}:
%\begin{equation}
%C= \sum_{i=0}^{\min\{d,k\}-1} \min\{(d-i)\beta, \alpha\}.
%\end{equation}
%A fundamental trade-off relationship between storage capacity $\alpha$ and repair bandwidth $\gamma = d\beta$ which suffice to transfer $\mathcal{M}$ amount of data is obtained.
%The result in \cite{dimakis2010network} also showed that, the repair bandwidth $\gamma$ is minimized when $d=n-1$, which means that newcomer need to contact every active node.

  \begin{figure}[!t]
	\centering
	\includegraphics[height=25mm]{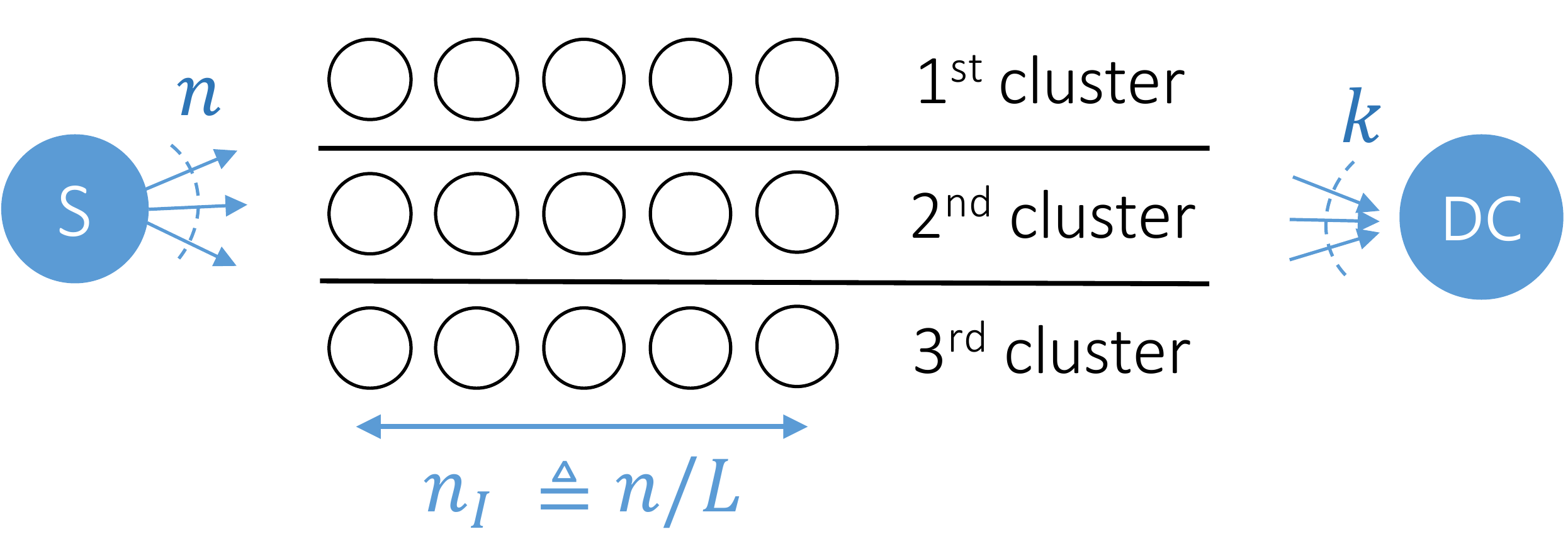}
	\caption{Clustered distributed storage system ($n=15,L=3$)}
	\label{Fig:Layered_DSS}
\end{figure}

\section{System Model}

\subsection{Clustered Distributed Storage System}

A distributed storage system with multiple clusters is shown in Fig. \ref{Fig:Layered_DSS}. 
Data from the source $\mathrm{S}$ is stored at $n$ nodes which are grouped into $L$ clusters. The number of nodes at each cluster is fixed and denoted as $n_I = n/L$. The storage size of each node is denoted as $\alpha$. 
When a node fails, a newcomer node is regenerated by contacting $d_I$ helper nodes within the same cluster, and $d_c$ helper nodes from other clusters.
%Each helper node reside in the same cluster transmits $\beta_I$ amount of data
The amount of data a newcomer node receives within the same cluster is 
$\gamma_I = d_I \beta_I$ (each node equally contributes $\beta_I$ ), and that from other clusters is $\gamma_c = d_c \beta_c$  (each node equally contributes $\beta_c$).
Fig. \ref{Fig:Repair Process in clustered DSS} illustrates an example of information flow graph representing the repair process in a clustered DSS.
%The repair process in clustered DSS is illustrated as a flow graph in Fig. \ref{Fig:Repair Process in clustered DSS}. 
We assume that $d_c$ and $d_I$ have the maximum possible values ($d_c = n-n_I, d_I = n_I - 1$), since this setting most efficiently utilizes the system resources: node storage size and repair bandwidth \cite{dimakis2010network}. 
A data collector $\mathrm{DC}$ contacts any $k$ out of $n$ nodes in the clustered DSS. 

  \begin{figure}[!t]
	\centering
	\includegraphics[height=30mm]{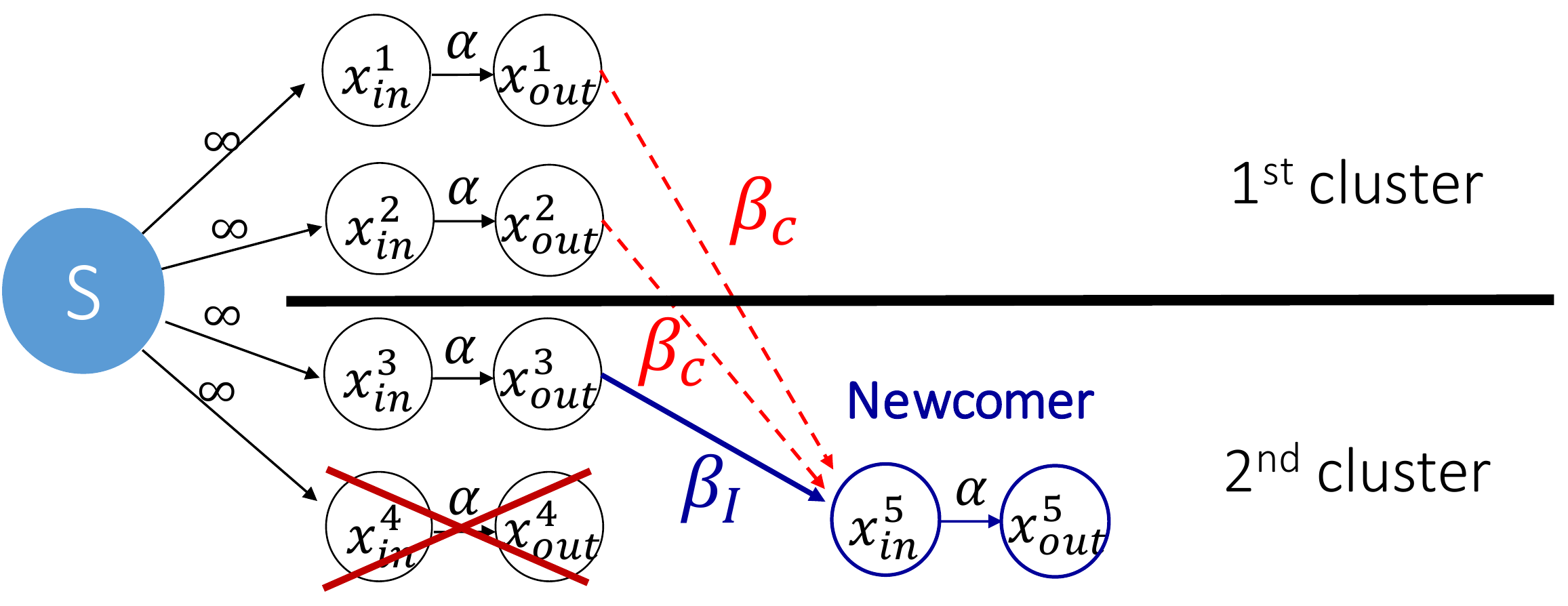}
	\caption{Repair process in clustered DSS ($n=4, L=2, d_I =1, d_c = 2$) }
	\label{Fig:Repair Process in clustered DSS}
\end{figure}

%  \begin{figure}[!t]
%    \centering
%    \includegraphics[height=35mm]{Repair_Model.pdf}
%    \caption{Repair Process in clustered DSS (2D structure)}
%    \label{Fig:Suggested System Model}
%\end{figure} 
%

\subsection{Problem Formulation}

Consider a clustered DSS with fixed $n, k, L$ values. In this model, we want to find the set of \textit{feasible} parameters ($\alpha, \gamma_I, \gamma_c$) which enables storing data of size $\mathcal{M}$. 
In order to find the feasible set, min-cut analysis on the information flow graph is required, similar to \cite{dimakis2010network}.
%The result of (possibly infinite) failure-repair process and $k$ node contacting at $\mathrm{DC}$ can be expressed as a flow graph $G$. 
Depending on the failure-repair process and $k$ nodes contacted by $\mathrm{DC}$, various information flow graphs can be obtained. 

Let $\mathcal{G}$ be the set of all possible flow graphs. Denote the graph with minimum min-cut as $G^*$.
Consider arbitrary information flow graph $G \in \mathcal{G}$.
Based on the max-flow min-cut theorem in \cite{ahlswede2000network}, 
the maximum information flow from the source to the data collector for $G$ is greater than equal to 
\begin{equation}\label{Eqn:capacity expression}
\mathcal{C}(\alpha, \gamma_I, \gamma_c) \triangleq \text{(min-cut of }G^* \text{)},
\end{equation}
which is called the \textit{capacity} of the system. 
In order to send data $\mathcal{M}$ from the source to the data collector, $\mathcal{C} \geq \mathcal{M}$ should be satisfied. 
Moreover, if $\mathcal{C} \geq \mathcal{M}$ is satisfied, there exists a linear network coding scheme \cite{ahlswede2000network} to store a file with size $\mathcal{M}$. 
Therefore, the set of $(\alpha, \gamma_I, \gamma_c)$ points which satisfies $\mathcal{C} \geq \mathcal{M}$ is \textit{feasible} in the sense of reliably storing the original file of size $\mathcal{M}$.

This paper first finds the min-cut minimizer $G^* \in \mathcal{G}$ and then obtains capacity $\mathcal{C} (\alpha, \gamma_I, \gamma_c)$, which is done in section \ref{Section:main results}.
Given that the typical intra-cluster communication bandwidth is larger than the cross-cluster bandwidth in real systems, 
we assume $\beta_I \geq \beta_c$ throughout the present paper.

\newtheorem{theorem}{Theorem}
\newtheorem{lemma}{Lemma}
\newtheorem{corollary}{Corollary}
\newtheorem{definition}{Definition}
\newtheorem{prop}{Proposition}

\section{Capacity of clustered DSS}\label{Section:main results}
In this section, a closed-form solution for capacity $\mathcal{C}(\alpha,  \gamma_I, \gamma_c)$ in (\ref{Eqn:capacity expression}) is obtained by specifying the min-cut minimizer $G^*$.
For a fixed $G \in \mathcal{G}$, let $\{x_{out}^{c_i} \}_{i=1}^{k}$ be the set of output nodes contacted by the data collector, ordered by topological sorting. Note that every directed acyclic graph can be topologically sorted \cite{bang2008digraphs}, where vertex $u$ is followed by vertex $v$ if there exists a directed edge from $u$ to $v$.
Depending on the selection of $k$ output nodes (among $n$ nodes) and ordering of the selected nodes, different kinds of flow graphs with possibly different min-cut values are obtained. Therefore, in order to obtain min-cut minimizing flow graph $G^*$, we need to specify $i)$ the optimal ordering method of the given $k$ output nodes and $ii)$ the optimal selection method of $k$ output nodes, which are stated as Lemmas 1 and 2, respectively.
A selection method is mathematically expressed as a selection vector $\bm{s}$, while an ordering method is expressed as an ordering vector $\bm{\pi}$, defined in the following subsection.

%However, in order to obtain $G^*$, we need some mathematical notations: selection vector, ordering vectors, and min-cut. Based on the mathematical notations, $G^*$ and $C(\alpha,  \gamma_I, \gamma_c)$ is specified in this section.

\subsection{Selection vector, ordering vector and min-cut}

%Here, we define the selection vector $\bm{s}$ and the ordering vector $\bm{\pi}$ which represents selection method of $k$ out-nodes and ordering method of given $k$ out-nodes, respectively. Using these definitions, min-cut value can be expressed as $c_{min}(\bm{s}, \bm{\pi})$. 

%Here, we define the selection vector $\bm{s}$.
\begin{definition}\label{Def:Selection vector}
Let arbitrary $k$ nodes are selected as the output nodes $\{x_{out}^{c_i}\}_{i=1}^k$. Label each cluster by the number of selected nodes in a descending order. In other words, the $1^{st}$ cluster contains a maximum number of selected nodes, and the $L^{th}$ cluster contains a minimum number of selected nodes. Under this setting, define the selection vector $\bm{s} = [s_1, s_2, \cdots, s_L]$ where $s_i$ is the number of selected nodes in the $i^{th}$ cluster. 
\end{definition}

  \begin{figure}[!t]
    \centering
    \includegraphics[height=25mm]{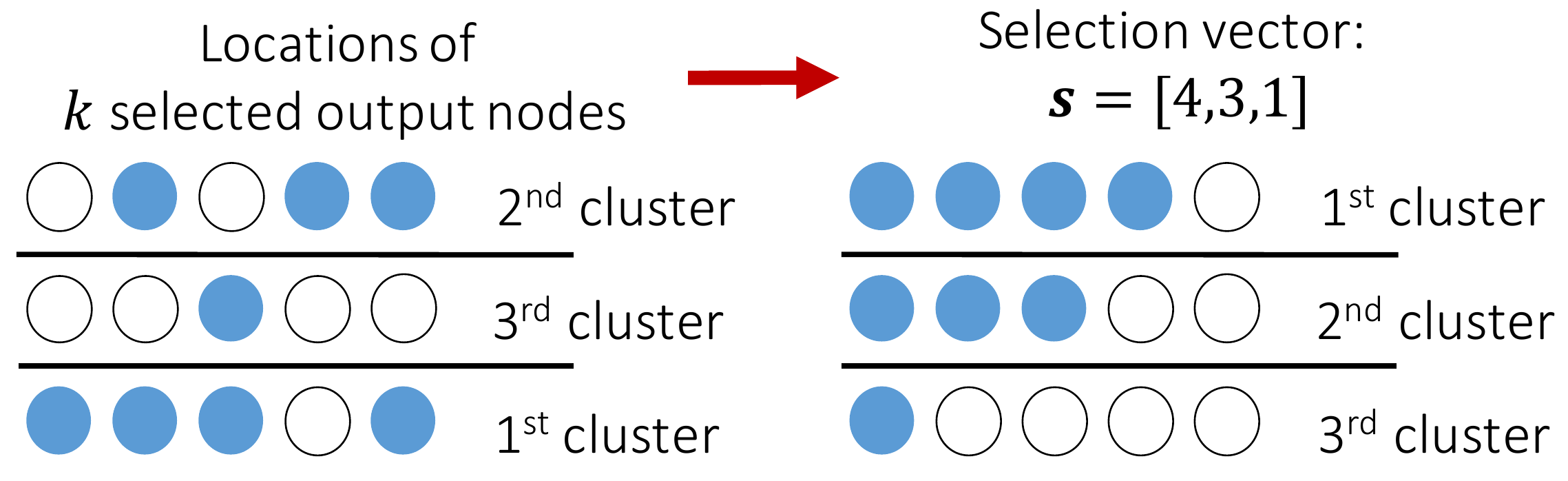}
    \caption{Obtaining the selection vector for given $k$ output nodes  ($n=15, k=8, L=3$)}
    \label{Fig:Selection Vector}
\end{figure}

For the selected $k$ output nodes, the corresponding selection vector $\bm{s}$ is illustrated in Fig. \ref{Fig:Selection Vector}. 
From the definition of selection vector, the set of possible selection vectors can be specified as follows.
%\begin{definition}\label{Def:Selection vector set}
%For given $n,k,L$ value, the possible set of selection vectors is
\begin{equation*}
\mathcal{S} \triangleq  \left\{\bm{s} =[s_1, \cdots, s_L] : 0 \leq s_i \leq n_I, s_{i+1} \leq s_{i}, \sum_{i=1}^{k} s_i = k \right\}
\end{equation*}
%\end{definition}

%  \begin{figure}[!t]
  %  \centering
    %\includegraphics[height=40mm]{set_of_selection_vector.pdf}
    %\caption{Surjection from $\mathcal{M}$ to  $\mathcal{S}$}
    %\label{Fig:Set of Selection Vector}
%\end{figure} 

% Original version (before data compression)
%Now, we can focus on the relationship between $k$-node selection method and corresponding selection vector $\bm{s}$, illustrated in Fig. \ref{Fig:Set of Selection Vector}. Consider $\mathcal{M}$, the set of all possible $k$-node selection methods among $n$ nodes. By using the definition \ref{Def:Selection vector}, every ${n\choose k}$ elements of $\mathcal{M}$ can be assigned to a selection vector $\bm{s} \in \mathcal{S}$, i.e., this is onto function from $\mathcal{M}$ to $\mathcal{S}$. 
%Note that even though ${n\choose k} $ number of different $k$-node selections exist, corresponding selection vector $\bm{s}$ only determines the min-cut value. To be specific, different $k$-node selection methods with same $\bm{s}$ vector ($m_1$ and $m_2$ in Fig. \ref{Fig:Set of Selection Vector})
%has identical set of possible flow graphs, which result in identical min-cut value. 
%Therefore, comparing min-cut values of all possible selection vectors $\bm{s}$ is enough, instead of comparing min-cut values of ${n\choose k}$ number of $k$-node selection methods.

% Edited version (explain why s vector is enough)
By using Definition \ref{Def:Selection vector}, every selection of $k$ nodes can be assigned to a selection vector $\bm{s} \in \mathcal{S}$.
Note that even though ${n\choose k}$ different selections exist, the min-cut value is only determined by the corresponding selection vector $\bm{s}$. To be specific, consider different selections $\varsigma_1, \varsigma_2$ assigned to the same $\bm{s}$ vector. 
Then, every possible flow graph originated from $\varsigma_1$ is isomorphic to an element of possible flow graphs originated from $\varsigma_2$, and vice versa. Thus, $\varsigma_1$ and $\varsigma_2$ have the identical set of min-cut values.
%To be specific,
%different selections which assigned to  same $\bm{s}$ vector 
%has identical set of possible flow graphs, which result in identical set of min-cut values. 
Therefore, comparing the min-cut values of all $\vert \mathcal{S} \vert$ possible selection vectors $\bm{s}$ is enough; it is not necessary to compare the min-cut values of ${n\choose k}$ selection methods.

%define ordering vector
Now, we define the ordering vector $\bm{\pi}$ for a given selection vector $\bm{s}$. 
\begin{definition}\label{Def:Ordering vector}
Let the locations of $k$ output nodes $\{x_{out}^{c_i}\}_{i=1}^k$ be fixed, with a corresponding selection vector $\bm{s} = [s_1, \cdots, s_L]$. Then, for arbitrary ordering of $k$ output nodes, define the ordering vector $\bm{\pi} = [\pi_1, \cdots, \pi_k]$ where $\pi_i$ is the index of cluster which contains $x_{out}^{c_i}$.
\end{definition}
For a given $\bm{s}$, the ordering vector $\bm{\pi}$ corresponding to an arbitrary ordering of $k$ nodes is illustrated in Fig. \ref{Fig:Ordering Vector}. 
In this figure (and the following figures in this paper), the number $i$ written inside each node means that the node is $x^{c_i}$.

From the definition, an ordering vector $\bm{\pi} \in \Pi(\bm{s})$ has $s_l$ components with value $l$, for all $l\in \{1, \cdots, L\}$. 
The set of possible ordering vectors can be specified as follows.
%\begin{definition}\label{Def:Ordering vector set}
%Consider a fixed selection vector $\bm{s} = [s_1, \cdots, s_L]$. Then, the set of possible ordering vectors is defined as
\begin{equation*}
\Pi(\bm{s}) = \left\lbrace \bm{\pi} = [\pi_1, \cdots, \pi_k] : \sum_{i=1}^{k} \mathds{1}_{\pi_i = l} = s_l \ \forall l \in \{1,\cdots,L\} \right\rbrace 
\end{equation*}
%\end{definition}
Here, $\mathds{1}_{\pi_i = l}$ is an indicator function which has value $1$ if $\pi_i = l$, and 0 otherwise.
Note that for given $k$ selected nodes, there exists $k! $ different ordering methods. However, the min-cut value is only determined by 
the corresponding ordering vector $\bm{\pi} \in \Pi(\bm{s})$ (based on flow graph analysis, similar to compressing ${n\choose k}$ selection methods to $\vert \mathcal{S} \vert$ selection vectors). 
Therefore, comparing the min-cut values of all possible ordering vectors $\bm{\pi}$ is enough; it is not necessary to compare the min-cut values of all $k!$ ordering methods.

  \begin{figure}[!t]
	\centering
	\includegraphics[height=25mm]{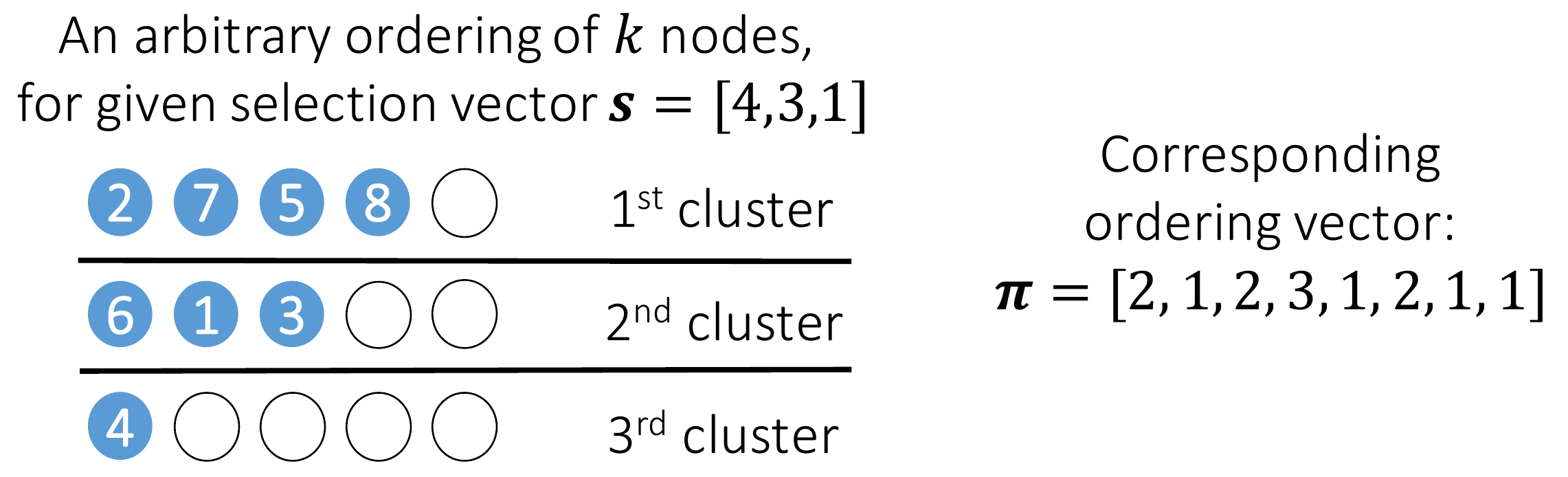}
	\caption{Obtaining the ordering vector for given an arbitrary order of $k$ output nodes ($n=15, k=8, L=3$)}
	\label{Fig:Ordering Vector}
\end{figure}

% Original version (before data compression)
%Now we can focus on the relationship between permutation of $k$ nodes and corresponding ordering vector $\bm{\pi}$, illustrated in Fig. \textcolor{red}{sth (at note 44 page - 4th)}. For a fixed selection vector $\bm{s}$, consider $\mathcal{P}(\bm{s})$, the set of all permutation of $k$ selected nodes.
%From the definition \ref{Def:Ordering vector}, every $k!$ elements of $\mathcal{P}(\bm{s})$ can be assigned to an ordering vector $\bm{\pi} \in \Pi(\bm{s})$.
%Note that different permutation method with same $\bm{\pi}$ vector ($p_1$ and $p_2$ in Fig. \textcolor{red}{sth (at note 44 page - 4th)}) has identical set of possible flow graphs, which result in identical min-cut value. Therefore, comparing min-cut values of all possible ordering vectors $\pi \in \Pi(\bm{s})$ is enough, instead of comparing min-cut values of $k!$ number of permutations.
%Finding min-cut minimizing flow graph $G^*$ can be graphically illustrated in Fig. \textcolor{red}{sth (at note 45 page - 5th)}. Depending on the selection/ordering vector pair ($\bm{s}, \bm{\pi}$), min-cut value $c_{min}(\bm{s}, \bm{\pi})$ changes (Here, $\bm{s} \in \mathcal{S}$ and $\bm{\pi} \in \Pi(\bm{s})$).  
%Therefore, all we need to do is finding optimal selection/ordering vector pair ($\bm{s}^*, \bm{\pi}^*$) which minimizes min-cut.

% Edited version (explain why pi vector is enough) 

  \begin{figure}[!t]
	\centering
	\includegraphics[height=45mm]{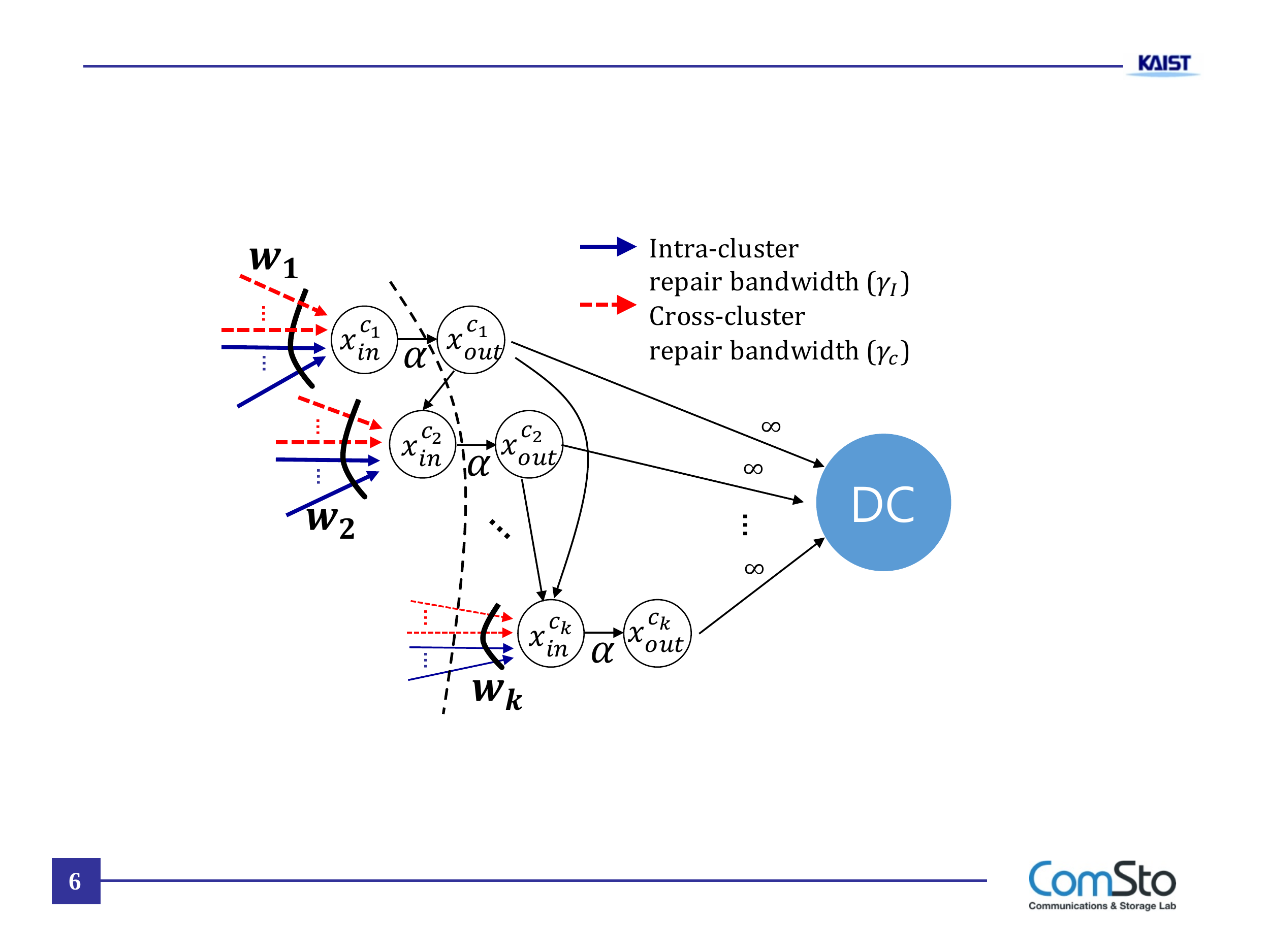}
	\caption{Information flow graph for obtaining min-cut in Proposition 1}
	\label{Fig:Flow graph for obtaining min-cut}
\end{figure}

Now, we express the min-cut value as a function of selection vector $\bm{s}$ and ordering vector $\bm{\pi}$. 

\begin{prop}\label{Prop:mincut expression}
	For a given $\bm{s}$, consider arbitrary ordering vector $\bm{\pi} \in \Pi(\bm{s})$, where $\bm{\pi} = [\pi_1, \cdots, \pi_k]$.
%	Consider arbitrary ordering vector $\bm{\pi} = [\pi_1, \cdots, \pi_k] \in \Pi(\bm{s})$ where $\bm{s} = [s_1, \cdots, s_L]$ is given.
	Then, the minimum min-cut value among possible flow graphs is
%	For given selection vector $\bm{s} = [s_1, \cdots, s_L]$ and ordering vector $\bm{\pi} = [\pi_1, \cdots, \pi_k]$, the min-cut value of corresponding flow graph is
	\begin{equation}\label{Eqn:min_cut}
	c_{min}(\bm{s, \pi}) = \sum_{i=1}^k \min\{\omega_i(\bm{\pi}), \alpha\}
	\end{equation}
	where
	\begin{align} 
	\omega_i(\bm{\pi}) &= a_i(\bm{\pi}) \beta_I + (n-i-a_i(\bm{\pi})) \beta_c \label{Eqn:weight vector}\\
	a_i(\bm{\pi}) &= n_I - 1 - \sum_{j=1}^{i-1} \mathds{1}_{\pi_j = \pi_i}. \label{Eqn:a_i}
	\end{align}
\end{prop}
Denote $\omega_i(\bm{\pi})$ in (\ref{Eqn:weight vector}) as the  \textit{$i^{th}$ weight value} of given ordering vector $\bm{\pi}$.
Here, we provide a proof sketch on Proposition \ref{Prop:mincut expression}. The formal proof will be given elsewhere.
Consider an arbitrary selection vector $\bm{s}$ and an ordering vector $\bm{\pi}$. 
Then, the $k$ nodes $\{x^{c_i}\}_{i=1}^k$ connected to the data collector are specified.
Consider an information flow graph $G$ illustrated in Fig. \ref{Fig:Flow graph for obtaining min-cut}, which satisfies the following: for every $i,j$ satisfying $1 \leq i < j \leq k$, there is a directed edge from $x_{out}^{c_i}$ to $x_{out}^{c_j}$.
In this graph, the sum of capacities of edges coming into $x_{in}^{c_i}$ (except those from $\{x_{out}^{c_t}\}_{t=1}^{i-1}$) turns out to be $\omega_i$ expressed in (\ref{Eqn:weight vector}).

Consider a set $C$ of edges generated as follows. For all $i$ satisfying $1\leq i \leq k$, compare values of $\omega_i$ and $\alpha$. If $\omega_i \leq \alpha$, include the edges coming into $x_{in}^{c_i}$ (except those from $\{x_{out}^{c_t}\}_{t=1}^{i-1}$) in the set $C$. Otherwise, include the edge from $x_{in}^{c_i}$ to $x_{out}^{c_i}$ to the set $C$.
Then, the generated set $C$ becomes a cut separating the source and the data collector. The sum of capacities of edges included in $C$ is obtained as $\sum_{i=1}^k \min\{\omega_i, \alpha\}$ in (\ref{Eqn:min_cut}).
Finally, any information flow graphs for given selection vector $\bm{s}$ and ordering vector $\bm{\pi}$ are shown to have min-cut values greater than or equal to RHS of (\ref{Eqn:min_cut}), which completes the proof.

We now state the useful property of $\omega_i$ defined in (\ref{Eqn:weight vector}).
\begin{prop}\label{Prop:weight vector}
	Let $n,k,L$ and selection vector $\bm{s}$ be fixed. Then, $\sum_{i=1}^k \omega_i(\bm{\pi})$ is constant irrespective of the ordering vector $\bm{\pi} \in \Pi(\bm{s})$. 
%	\begin{equation} \label{eqn:sum of weight}
%	\Omega(\bm{s}) \triangleq \sum_{i=1}^k \omega_i
%	\end{equation}
\end{prop}

\begin{proof}
Consider a fixed selection vector $\bm{s}= [s_1, \cdots, s_L]$.
For an arbitrary ordering vector $\bm{\pi} \in \Pi(\bm{s})$, let $b_i(\bm{\pi})= n-i-a_i(\bm{\pi})$ where $a_i(\bm{\pi})$ is at (\ref{Eqn:a_i}). 
For simplicity, we denote $a_i(\bm{\pi})$ and $b_i(\bm{\pi})$ as $a_i$ and $b_i$, respectively.
Then, 
\begin{equation}\label{Eqn:proof of proposition 2}
c_0 \triangleq \sum_{i=1}^{k}(a_i + b_i) = \sum_{i=1}^{k}(n-i),
\end{equation}
where $c_0$ is constant.
Note that $\sum_{i=1}^{k}a_i = k(n_I - 1) - \sum_{i=1}^{k} \sum_{j=1}^{i-1} \mathds{1}_{\pi_j = \pi_i}$. Also, from the definition of $\Pi(\bm{s})$, the ordering vector $\bm{\pi}$ has $s_l$ components with value $l$, for all $l\in \{1, \cdots, L\}$. 
Let $I_l \triangleq \{ i \in \{ 1, \cdots, k\} : \pi_i = l \}$ for $l=1,\cdots, L$. Then, 
$\sum_{i\in I_l} \sum_{j=1}^{i-1} \mathds{1}_{\pi_j = \pi_i} = 0 + 1 + \cdots + (s_l-1)$.
Therefore,
\begin{equation*}
c_1 \triangleq \sum_{i=1}^{k} \sum_{j=1}^{i-1} \mathds{1}_{\pi_j = \pi_i} = \sum_{l=1}^{L} \sum_{i\in I_l} \sum_{j=1}^{i-1} \mathds{1}_{\pi_j = \pi_i} = \sum_{l=1}^{L} \sum_{t=0}^{s_l-1} t,
\end{equation*}
where $c_1$ is constant.
Thus, $\sum_{i=1}^{k}a_i  = k(n_I - 1) -c_1$ is also a constant. From (\ref{Eqn:proof of proposition 2}), we get $\sum_{i=1}^{k}b_i  = c_0 - \sum_{i=1}^{k}a_i$, which is another constant.
Therefore, $\sum_{i=1}^k \omega_i = (\sum_{i=1}^k a_i) \beta_I + (\sum_{i=1}^k b_i) \beta_c$ is constant for every ordering vector $\bm{\pi} \in \Pi(\bm{s})$.
\end{proof}

\subsection{Two Lemmas Supporting the Main Theorem}

Now, we state our first main Lemma, which specifies the optimal ordering vector $\bm{\pi}$ which minimizes $c_{min}(\bm{s, \pi})$ for an arbitrary selection vector $\bm{s}$.

% Lemma 1
\begin{lemma} \label{Lemma:optimal ordering vector}
Let $\bm{s} \in \mathcal{S}$ be an arbitrary selection vector. Then, a vertical ordering vector $\bm{\pi}_v$ obtained by Algorithm \ref{Alg:vertical ordering} minimizes the min-cut.
In other words, $c_{min}(\bm{s}, \bm{\pi}_v)  \leq c_{min}(\bm{s}, \bm{\pi})$ holds for arbitrary $\bm{\pi} \in \Pi(\bm{s})$.
\end{lemma}

% Algorithm obtaining vertical ordering vector
\begin{algorithm}[t]
\caption{Generate vertical ordering $\bm{\pi}_v  $}
\label{Alg:vertical ordering}
\begin{algorithmic}
\STATE \textbf{Input:} $\bm{s} = [s_1, \cdots, s_L]$
\STATE \textbf{Output:} $\bm{\pi}_v = [\pi_1, \cdots, \pi_k]$
\STATE Initialization: $l \leftarrow 1$
\FOR{$i=1$ to $k$}
\IF{$s_l = 0$}
\STATE $l \leftarrow 1$
\ELSE
\STATE $\pi_i \leftarrow l$; $s_i \leftarrow s_i - 1$; $l \leftarrow l + 1$
\ENDIF
\ENDFOR
\end{algorithmic}
\end{algorithm}

  \begin{figure}[t]
	\centering
	\includegraphics[height=25mm]{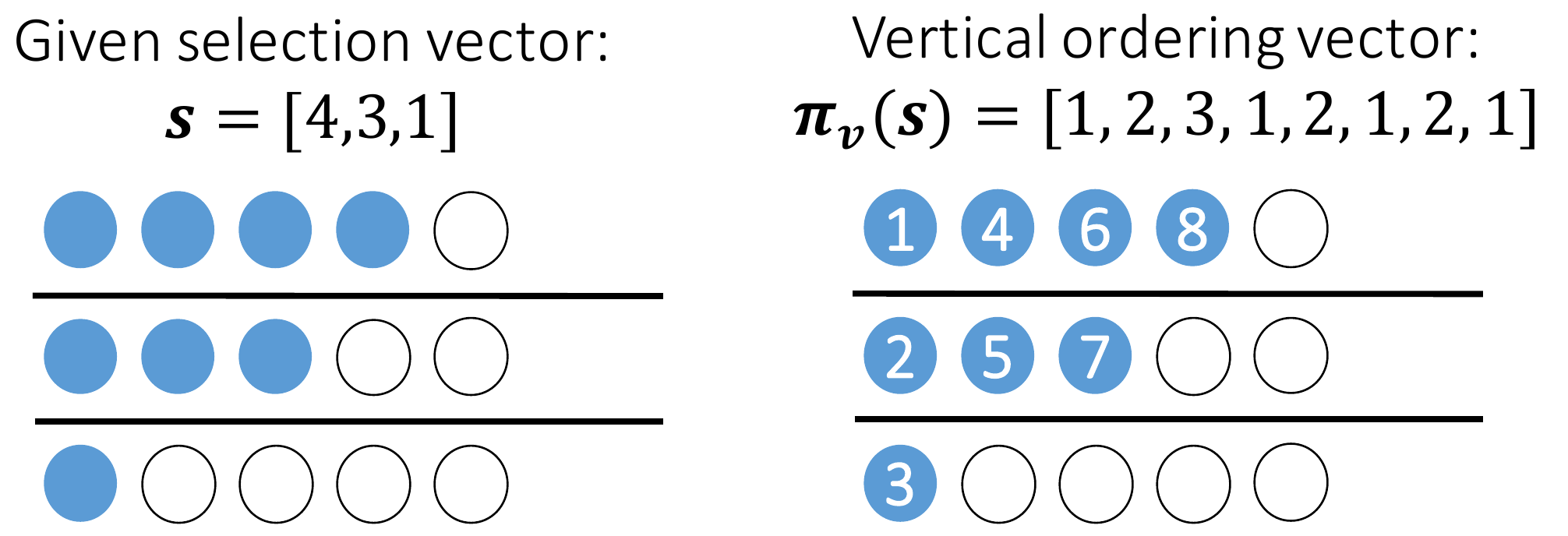}
	\caption{The vertical ordering vector $\bm{\pi}_v $ for given selection vector $\bm{s}=[4,3,1]$ (for $n=15, k=8, L=3$ case) }
	\label{Fig:vertical ordering}
\end{figure}

The vertical ordering vector is illustrated in Fig. \ref{Fig:vertical ordering}, for a given selection vector $\bm{s}$ as an example.
For $\bm{s}=[4,3,1]$, Algorithm 1 produces the corresponding vertical ordering vector $\bm{\pi}_v = [1,2,3,1,2,1,2,1]$. Note that the order of $k=8$ output nodes is  illustrated in Fig. \ref{Fig:vertical ordering}, as the numbers inside each node. 
Although the vertical ordering vector $\bm{\pi}_v$ depends on the selection vector $\bm{s}$, we use simplified notation $\bm{\pi}_v$.

% Proof sketch of Lemma 1
Due to space limitation, here we just sketch the proof of  Lemma \ref{Lemma:optimal ordering vector}. The formal proof will be given elsewhere. 
%First, if $\beta_I = \beta_c = \beta$, then $\omega_i = (n-i)\beta$ holds so that min-cut value is constant irrespective of ordering vector $\bm{\pi} \in \Pi(\bm{s})$. Therefore, rest of the proof sketch deals with  $\beta_I > \beta_c$ case.
Let a selection vector $\bm{s}$ be given.
Consider a running sum $S_t(\bm{\pi}) = \sum_{i=1}^{t} \omega_i (\bm{\pi})$ for an arbitrary ordering vector $\bm{\pi} \in \Pi(\bm{s})$.
Suppose there exists a running sum maximizer $\bm{\pi^*}$ which satisfies $S_t(\bm{\pi}) \leq S_t(\bm{\pi^*}) \ $
for every $t \in \{1, \cdots, k\}$ and every $\bm{\pi} \in \Pi(\bm{s})$.
%\forall t, \forall \bm{\pi} \in \Pi(\bm{s})$.
Then, $\bm{\pi^*}$ turns out to be the min-cut minimizer. 
%For arbitrary given ordering vector $\bm{\pi}$, imagine plotting the $\omega_i$ values for $i=1,\cdots, k$ as in Fig. \ref{Fig:omega_I_point}. Note that $\sum_{i=1}^{k}\omega_i$ is constant irrespective of ordering vector $\bm{\pi}$ by proposition \ref{Prop:weight vector}.
%Since min-cut is sum of $\min\{\omega_i,\alpha\}$, 
%%putting maximum number of $(i,\omega_i)$ points above $\alpha$ 
%maximizing $\sum_{i=1}^k (\omega_i - \min\{\omega_i,\alpha\}) $ minimizes min-cut. Therefore, running sum maximizer which has the most \textit{slowly} decaying $\{\omega_i\}$ sequence minimizes min-cut.
Showing that the vertical ordering vector $\bm{\pi}_v$ is the unique running sum maximizer completes the proof, which is omitted here.

%\begin{figure}[!t]
%	\centering
%	\includegraphics[height=30mm]{omega_I_point.pdf}
%	\caption{Plot of $(i, \omega_i)$ points for $i=1, \cdots, k$}
%	\label{Fig:omega_I_point}
%\end{figure} 

Before stating our second main Lemma, here we define a special selection vector called the \textit{horizontal selection vector}. 
\begin{definition}\label{Def:horizontal selection vector}
The horizontal selection vector $\bm{s}_h = [s_1, \cdots, s_L] \in \mathcal{S}$ is defined as:
\begin{equation}
s_i =
\begin{cases}
n_I,    &    i \leq \floor{\frac{k}{n_I}} \nonumber\\
k-\floor{\frac{k}{n_I}}n_I,    & i = \floor{\frac{k}{n_I}} + 1 \nonumber\\
0		&     i > \floor{\frac{k}{n_I}} + 1.
\end{cases}
\end{equation} 
\end{definition}

  \begin{figure}[t]
	\centering
	\includegraphics[height=25mm]{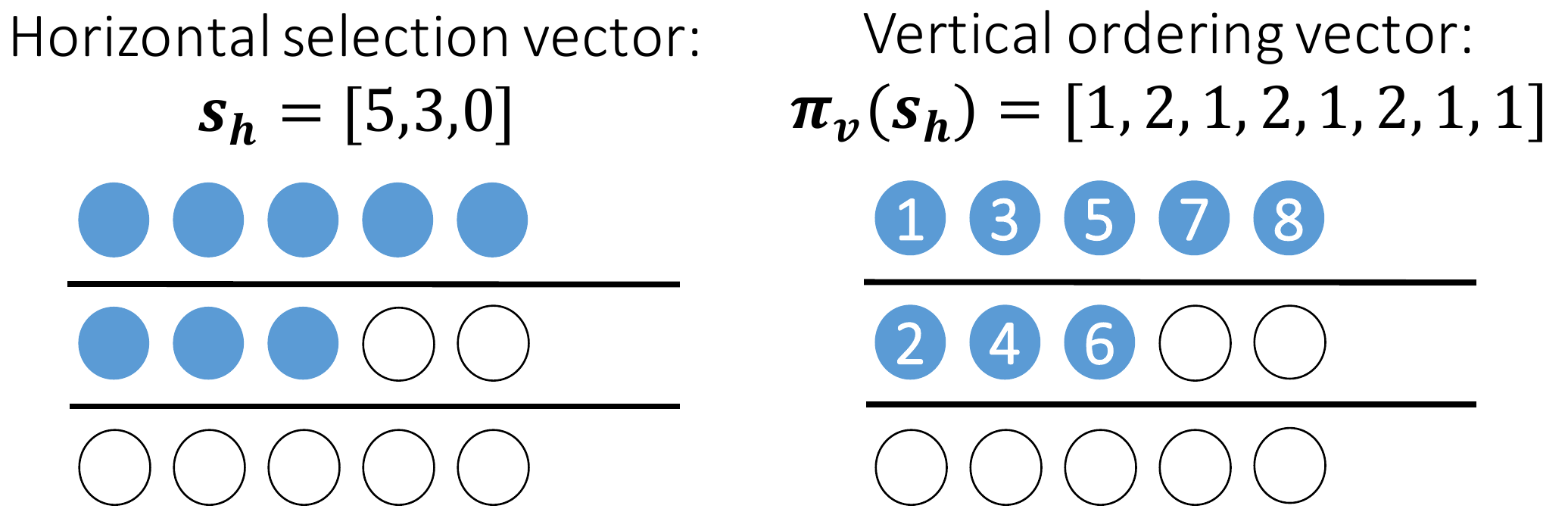}
	\caption{The optimal selection vector $\bm{s}_h$  and the optimal ordering vector $\bm{\pi}_v $ (for $n=15, k=8, L=3$ case)}
	\label{Fig:horizontal selection}
\end{figure}

The graphical illustration of the horizontal selection vector is on the left side of Fig. \ref{Fig:horizontal selection}, in the case of $n=15, k=8, L=3$.
Now, our second main Lemma states that the horizontal selection vector minimizes the min-cut. 

% Lemma 2
\begin{lemma} \label{Lemma:optimal selection vector}
When the vertical ordering $\bm{\pi}_v$ is selected,
 the horizontal selection vector $\bm{s}_h$ minimizes the min-cut.
In other words, $c_{min}(\bm{s}_h, \bm{\pi}_v) \leq c_{min}(\bm{s}, \bm{\pi}_v ) \ \forall \bm{s}\in \mathcal{S}$.
\end{lemma}
The proof of this Lemma will be given elsewhere, but it is based on following logic. 
For the given horizontal selection vector $\bm{s}_h$ and the corresponding vertical ordering vector $\bm{\pi}_v$, denote the $i^{th}$ weight value in (\ref{Eqn:weight vector}) as $\omega_i^*$.
Similarly, for an arbitrary selection vector $\bm{s}$ and the corresponding vertical ordering vector $\bm{\pi}_v$, denote the $i^{th}$ weight value as $\omega_i$.
Then, it can be shown that
$\omega_i^* \leq \omega_i$ for $i=1, \cdots, k$. Therefore, the horizontal selection minimizes the min-cut directly from (\ref{Eqn:min_cut}).

%the analysis of $(i, \omega_i)$ plot in Fig.\ref{Fig:omega_I_point}. Under constraint of vertical ordering, $(i, \omega_i)$ plot changes depending on the choice of $\bm{s}$.

\subsection{Main Theorem: Capacity of Clustered DSS}

Now, we state our main result in the form of a theorem which offers a closed-form solution for the capacity of the clustered DSS.
Note that setting $L=1$ or $\beta_I = \beta_c$ reduces to the capacity of non-clustered DSS obtained in \cite{dimakis2010network}.

\begin{theorem} \label{Thm:Capacity of clustered DSS}
Consider a $\beta_I \geq \beta_c$ case. 
The capacity of the clustered distributed storage system with parameters $(n,k,L,\alpha, \gamma_I, \gamma_c)$ is 
\begin{equation}\label{Eqn:Capacity of clustered DSS}
\mathcal{C}(\alpha, \gamma_I, \gamma_c) = \sum_{i=1}^{n_I} \sum_{j=1}^{g(i)} \min \{\alpha, x(i)\gamma_I + y(i,j) \gamma_c \},
\end{equation}
where
\begin{align}
x(i) &= \frac{n_I - i}{n_I - 1} \nonumber\\
y(i,j) &= \frac{n-(n_I-i)-\sum_{m=1}^{i-1}g(m) - j}{n - n_I}\nonumber\\
g(i) &=
\begin{cases}
\floor{\frac{k}{n_I}} + 1, & i \leq mod(k,n_I) \nonumber\\
\floor{\frac{k}{n_I}}, & otherwise.
\end{cases}
\end{align}
\end{theorem}
\begin{proof}
	From Lemmas \ref{Lemma:optimal ordering vector} and \ref{Lemma:optimal selection vector}, we have 
	\begin{equation}
	\forall \bm{s} \in \mathcal{S}, \forall \bm{\pi} \in \Pi(\bm{s}), \ c_{min}(\bm{s}_h, \bm{\pi}_v) \leq c_{min}(\bm{s}, \bm{\pi}).
	\end{equation}
	Therefore, from the explanation on $G^*$ in the beginning of Section \ref{Section:main results},
	the min-cut minimizing flow graph $G^*$ is generated by selecting $k$ output nodes $\{x_{out}^{c_i}\}_{i=1}^k$ by the \textit{horizontal selection vector} $\bm{s}_h$ and ordering these nodes by the \textit{vertical ordering} $\bm{\pi}_v$, as illustrated in Fig. \ref{Fig:horizontal selection}.
	The min-cut value of $G^*$, or $c_{min}(\bm{s}_h, \bm{\pi}_v)$, is summarized as (\ref{Eqn:Capacity of clustered DSS}).
\end{proof}

\subsection{Relationship between $\mathcal{C}$ and $\kappa = \beta_c/\beta_I$}

In this subsection, we analyze the capacity of a clustered DSS as a function of a new important parameter $\kappa =  \beta_c/\beta_I$: the cross-cluster repair burden compared to the intra-cluster repair burden.
In Fig. \ref{Fig:capacity_versus_kappa_plot}, the capacity is plotted as a function of $\kappa$. The total repair bandwidth can be expressed as 
\begin{align}\label{Eqn:capacity versus kappa}
\gamma &= \gamma_I + \gamma_c = (n_I-1)\beta_I + (n-n_I)\beta_c \nonumber\\
&= \{n_I - 1 + (n-n_I)\kappa\} \beta_I.
\end{align}
Using this expression, the capacity is expressed as 
\begin{equation*} 
\mathcal{C}(\kappa) = \sum_{i=1}^{n_I} \sum_{j=1}^{g(i)} \min \{\alpha, \frac{(n-n_I)y(i,j)\kappa + (n_I-1)x(i) }{(n-n_I)\kappa + n_I - 1}   \gamma \}.
\end{equation*} 
For fair comparison on various $\kappa$ values, 
($n, k, L, \alpha, \gamma$) values are 
%the node storage size $\alpha$ and total repair bandwidth $\gamma$ in (\ref{Eqn:capacity versus kappa}) is 
fixed for calculating the capacity. 
The capacity shows an increasing function of $\kappa$ as in Fig. \ref{Fig:capacity_versus_kappa_plot}. This implies that for given resources $\alpha$ and $\gamma$, allowing a larger $\beta_c$ (until it reaches $\beta_I$) is always beneficial, in terms of storing a larger file. For example, under the setting in Fig. \ref{Fig:capacity_versus_kappa_plot}, allowing $\beta_c = \beta_I$ (i.e., $\kappa = 1$) can store $\mathcal{M}=48$, while setting $\beta_c = 0$ (i.e., $\kappa = 0$) cannot achieve the same level of storage.
This result is consistent with the previous work on asymmetric repair in \cite{ernvall2013capacity}, which proved that the symmetric repair maximized the capacity. 
Therefore, when the total communication amount $\gamma$ is fixed, the reduction of the storage capacity is the cost we need to pay in order to reduce the communication burden $\beta_c$ across different clusters.

\begin{figure}[t]
	\centering
	\includegraphics[height=45mm]{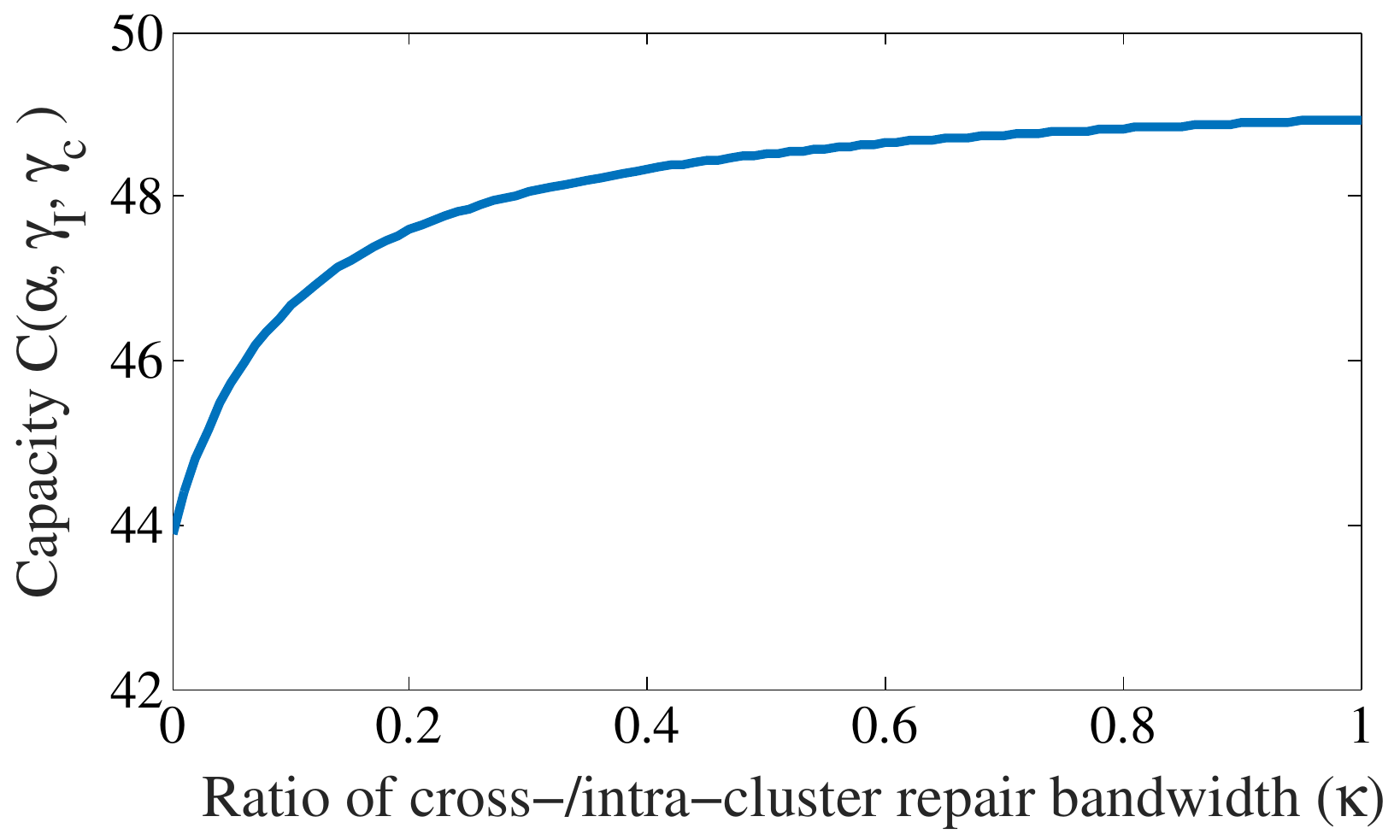}
	\caption{Capacity as a function of $\kappa$ ($n=100, k=85, L=10, \alpha = 1, \gamma = 1$)}
	\label{Fig:capacity_versus_kappa_plot}
\end{figure}

%The capacity value is also analyzed as a function of $L$, illustrated in  Fig. \ref{Fig:capacity_versus_L_plot}.
%\textcolor{red}{Delete? Not mathematically analyzed. }
%
%
%\begin{figure}[t]
%	\centering
%	\includegraphics[height=60mm]{capacity_versus_L_plot.pdf}
%	\caption{Capacity as a function of $L$ ($n=100, k=85, \alpha = 1, \gamma = 1$)}
%	\label{Fig:capacity_versus_L_plot}
%\end{figure}

\section{Discussion on Feasible ($\alpha, \gamma_I, \gamma_c$)}\label{Section:analysis on feasible points}

In the previous section, we obtained the capacity of the clustered DSS. 
This section analyzes the feasible ($\alpha, \gamma_I, \gamma_c$) points which satisfy $\mathcal{C}(\alpha, \gamma_I, \gamma_c) \geq \mathcal{M}$ for a given file size $\mathcal{M}$.
Here, the behavior of feasible points are analyzed in two different perspectives.
First, we focus on the $\gamma_c = 0$ case, 
which enables the local repair within each cluster.
%where a failed node is fully repaired by helper nodes in the same cluster only. 
Second, we analyze the trade-off relation between $\gamma_I$ and $\gamma_c$.

\subsection{Intra-Cluster Repairable Condition ($\gamma_c = 0$)}

Under the constraint of zero cross-cluster repair bandwidth, a closed-form solution for the feasible points ($\alpha, \gamma_I$) can be obtained as the following corollary. 
This corollary can be proved by substituting $\gamma_c = 0$ to (\ref{Eqn:Capacity of clustered DSS}), and obtaining an equivalent condition for $(\alpha, \gamma_I)$ which satisfies $\mathcal{C}(\alpha, \gamma_I, 0) \geq \mathcal{M}$. 
A detailed proof will be given elsewhere.
%Detailed proof is similar to Dimakis' proof of Theorem 1 in \cite{dimakis2010network}, which is omitted due to space limitation.

\begin{corollary}\label{Corollary:gamma_c=0 case}
Consider a clustered DSS %$\mathcal{D}(n,k,L)$
for storing data $\mathcal{M}$.
Then, for any $\gamma_I \geq \gamma_I^*(\alpha)$, the cross-cluster repair bandwidth can be reduced to zero, i.e., $\mathcal{C}(\alpha, \gamma_I, 0) \geq \mathcal{M}$, while it is impossible to reduce cross-cluster repair bandwidth to zero when $\gamma_I < \gamma_I^*(\alpha)$. The threshold function $\gamma_I^*(\alpha)$ can be obtained as:
	\begin{equation*} \label{Lemma3 result}
	\gamma_I^*(\alpha) = 
	\begin{cases}
	\frac{M-\delta_t\alpha}{\epsilon_t}, &  
	\ \alpha \in [\frac{M}{\epsilon_t/b_t + \delta_t} ,\frac{M}{\epsilon_{t-1}/b_{t-1} + \delta_{t-1}} )
	,\\
	& \ \ \ \ \ \ \ \ \ \ \ \ \ \ \ \ \  1 \leq t \leq n_I - 2\\
	 	\frac{M}{\epsilon_0 }, & \ \alpha \geq \frac{M}{\epsilon_0} 
	\end{cases}
	\end{equation*}
	where 
	\begin{align*}
	\epsilon_t &\triangleq \floor{\frac{k}{n_I}} \sum_{i=t}^{n_I-1}b_i + \sum_{i=t}^{k-\floor{\frac{k}{n_I}}n_I-1}b_i \\
	\delta_t &\triangleq
	\begin{cases}
	(\floor{\frac{k}{n_I}} + 1)t,  & \ 0 \leq t \leq k-\floor{\frac{k}{n_I}}n_I\\
	k-\floor{\frac{k}{n_I}}(n_I - t), & \ k-\floor{\frac{k}{n_I}}n_I < t \leq n_I - 1
	\end{cases}\\
	b_t &\triangleq 1-t/(n_I - 1).
	\end{align*}
\end{corollary}

\begin{figure}[t]
	\centering
	\includegraphics[height=55mm]{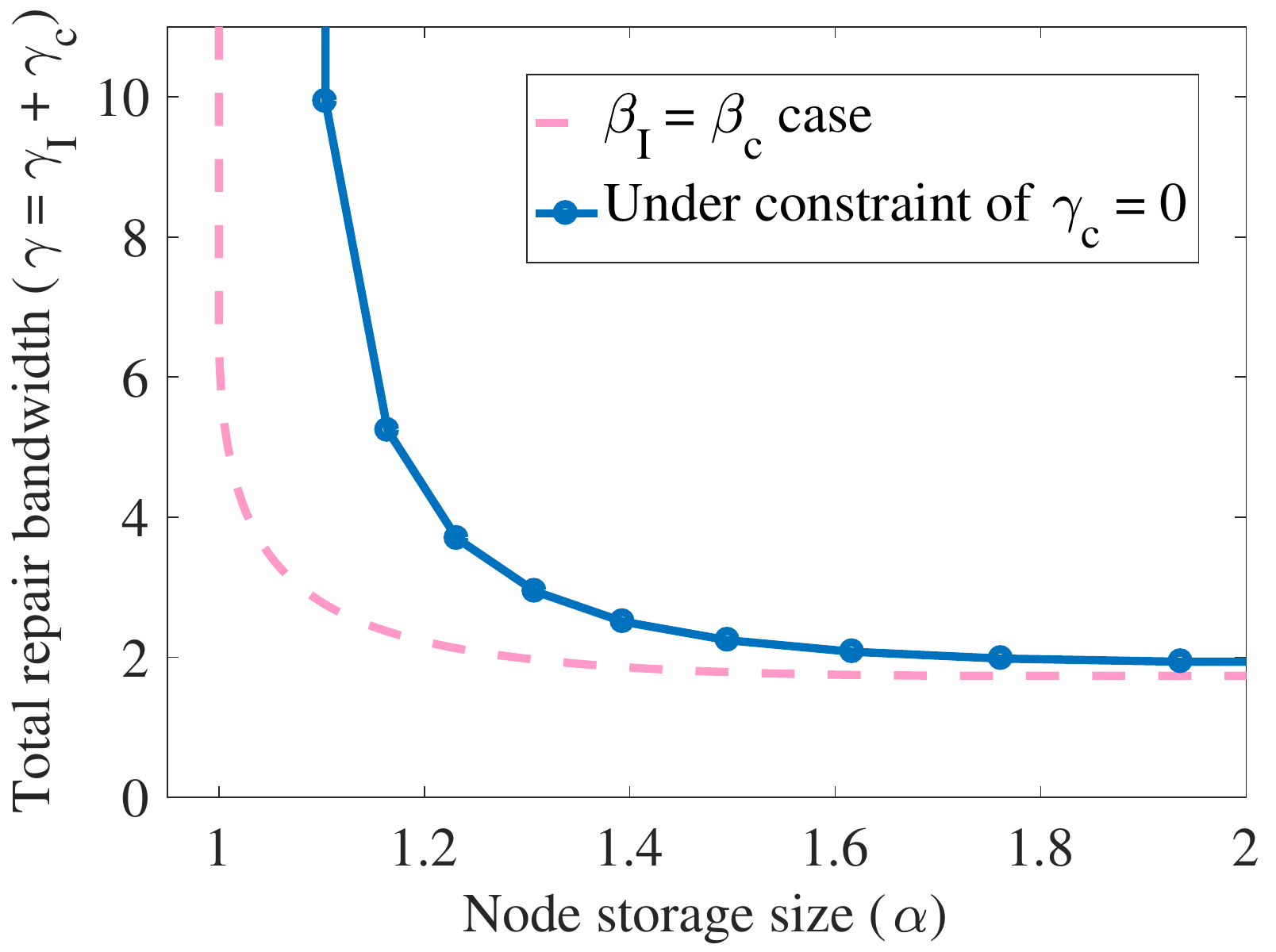}
	\caption{Optimal tradeoff between node storage size $\alpha$ and total repair bandwidth $\gamma$  ($n=100, k=85, L=10, \mathcal{M} = 85$)}
	\label{Fig:tradeoff_zero_gamma_c}
\end{figure}

An example for the tradeoff result of Corollary \ref{Corollary:gamma_c=0 case} is illustrated in Fig. \ref{Fig:tradeoff_zero_gamma_c}.
The plot for the $\beta_I = \beta_c$ case is illustrated as the dotted line. The solid line
represents the $\gamma_c = 0$ case, where in this case $\gamma = \gamma_I$. This shows that the cross-cluster repair bandwidth can be reduced to zero with extra resources ($\alpha$ and $\gamma_I$), where the set of feasible pairs of these resources is specified in Corollary \ref{Corollary:gamma_c=0 case}.
Note that $\gamma_c = 0$ means that the \textit{local repair in each cluster is possible}.

\subsection{Trade-off between $\gamma_I$ and $\gamma_c$}

In this subsection, we focus on the following problem: Given $\alpha$, what is the feasible set of $(\gamma_I, \gamma_c)$ pairs which satisfy $\mathcal{C}(\alpha, \gamma_I, \gamma_c) = \mathcal{M}$?
The sets of feasible $(\gamma_I, \gamma_c)$ pairs for various $\alpha$ values is illustrated in Fig. \ref{Fig:gamma_c_gamma_I}.
Note that the plots in the figure shows a decreasing function of $\alpha$ for a given $\gamma_I$. This has been obtained from  the capacity expression of  (\ref{Eqn:Capacity of clustered DSS}).

For small $\alpha$ values ($\mathcal{M}/k$ or $1.01 \times \mathcal{M}/k$ in the figure), the cross-cluster repair bandwidth cannot be reduced irrespective of the intra-cluster repair bandwidth $\gamma_I$. 
For sufficiently large $\alpha$ values, trade-off between $\gamma_c$ and $\gamma_I$ can be observed, with $\gamma_c$ settling to zero 
for large $\gamma_I$s as $\alpha$ increases.

\begin{figure}[!t]
	\centering
	\includegraphics[height=60mm]{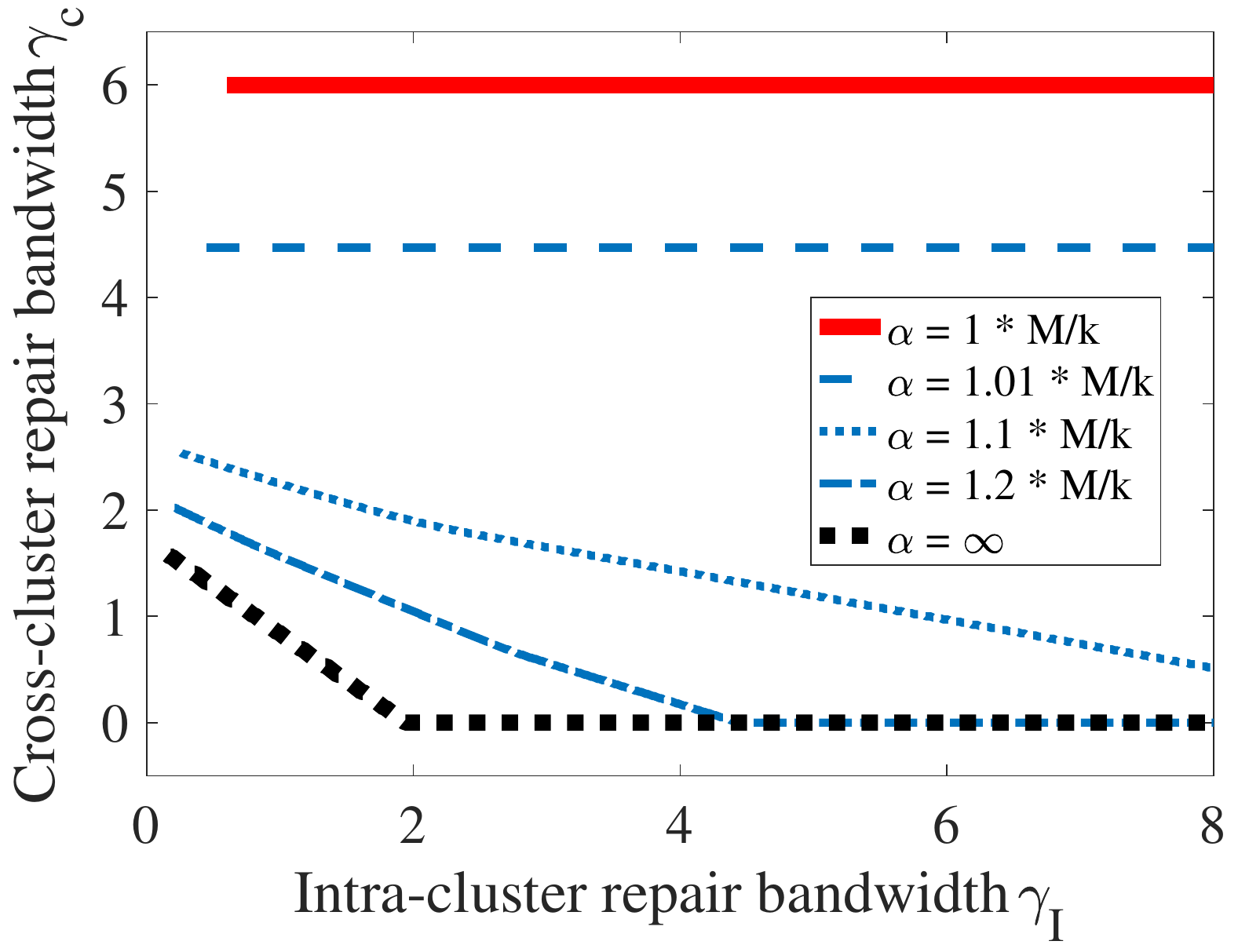}
	\caption{Optimal tradeoff between intra-/cross-cluster repair bandwidth ($n=100, k=85, L=10, \mathcal{M} = 85$)}
	\label{Fig:gamma_c_gamma_I}
\end{figure}

\section{Conclusion}

This paper considered a practical distributed storage system where storage nodes are dispersed  into several clusters. Noticing that the traffic burdens of intra- and cross-cluster communication are different, a new system model for a clustered distributed storage system is suggested. Based on the min-cut analysis of information flow graph, the storage capacity of the suggested model is obtained in a closed-form, as a function of three main resources: node storage size $\alpha$, intra-cluster repair bandwidth $\gamma_I$ and cross-cluster repair bandwidth $\gamma_c$. It is shown that the asymmetric repair ($\beta_I > \beta_c$) degrades the capacity, which is the cost to pay for reducing cross-cluster repair burden. 
Moreover, local repair ($\gamma_c = 0$) is shown to be possible at the expense of extra resources ($\alpha$ and $\gamma_I$), where the amounts of required extra resources are specified in a mathematical form.
Finally, the feasible ($\gamma_I, \gamma_c$) pairs showed a clear trade-off relation for large enough storage size $\alpha$. 

%Under constraint of arbitrary fixed total repair bandwidth $\gamma = \gamma_I + \gamma_c$, asymmetric repair ($\beta_I = \beta_c$) maximized the capacity, coinciding with the result of previous works. 

% appendices
%\appendices
%% Appendix A
%\section{Proofs of ...}
%
%\subsection{Proof of ...}
%
%\subsection{Proof of ...}

\bibliographystyle{IEEEtran}
\bibliography{IEEEabrv,ICC2017}

\begin{thebibliography}{10}
\providecommand{\url}[1]{#1}
\csname url@samestyle\endcsname
\providecommand{\newblock}{\relax}
\providecommand{\bibinfo}[2]{#2}
\providecommand{\BIBentrySTDinterwordspacing}{\spaceskip=0pt\relax}
\providecommand{\BIBentryALTinterwordstretchfactor}{4}
\providecommand{\BIBentryALTinterwordspacing}{\spaceskip=\fontdimen2\font plus
\BIBentryALTinterwordstretchfactor\fontdimen3\font minus
  \fontdimen4\font\relax}
\providecommand{\BIBforeignlanguage}[2]{{%
\expandafter\ifx\csname l@#1\endcsname\relax
\typeout{** WARNING: IEEEtran.bst: No hyphenation pattern has been}%
\typeout{** loaded for the language `#1'. Using the pattern for}%
\typeout{** the default language instead.}%
\else
\language=\csname l@#1\endcsname
\fi
#2}}
\providecommand{\BIBdecl}{\relax}
\BIBdecl

\bibitem{ghemawat2003google}
S.~Ghemawat, H.~Gobioff, and S.-T. Leung, ``The google file system,'' in
  \emph{ACM SIGOPS operating systems review}, vol.~37, no.~5.\hskip 1em plus
  0.5em minus 0.4em\relax ACM, 2003, pp. 29--43.

\bibitem{bhagwan2004total}
R.~Bhagwan, K.~Tati, Y.~Cheng, S.~Savage, and G.~M. Voelker, ``Total recall:
  System support for automated availability management.'' in \emph{NSDI},
  vol.~4, 2004, pp. 25--25.

\bibitem{dabek2004designing}
F.~Dabek, J.~Li, E.~Sit, J.~Robertson, M.~F. Kaashoek, and R.~Morris,
  ``Designing a dht for low latency and high throughput.'' in \emph{NSDI},
  vol.~4, 2004, pp. 85--98.

\bibitem{rhea2003pond}
S.~C. Rhea, P.~R. Eaton, D.~Geels, H.~Weatherspoon, B.~Y. Zhao, and
  J.~Kubiatowicz, ``Pond: The oceanstore prototype.'' in \emph{FAST}, vol.~3,
  2003, pp. 1--14.

\bibitem{dimakis2010network}
A.~G. Dimakis, P.~B. Godfrey, Y.~Wu, M.~J. Wainwright, and K.~Ramchandran,
  ``Network coding for distributed storage systems,'' \emph{IEEE Transactions
  on Information Theory}, vol.~56, no.~9, pp. 4539--4551, 2010.

\bibitem{ahlswede2000network}
R.~Ahlswede, N.~Cai, S.-Y. Li, and R.~W. Yeung, ``Network information flow,''
  \emph{IEEE Transactions on information theory}, vol.~46, no.~4, pp.
  1204--1216, 2000.

\bibitem{ford2010availability}
D.~Ford, F.~Labelle, F.~I. Popovici, M.~Stokely, V.-A. Truong, L.~Barroso,
  C.~Grimes, and S.~Quinlan, ``Availability in globally distributed storage
  systems.'' in \emph{OSDI}, 2010, pp. 61--74.

\bibitem{huang2012erasure}
C.~Huang, H.~Simitci, Y.~Xu, A.~Ogus, B.~Calder, P.~Gopalan, J.~Li, and
  S.~Yekhanin, ``Erasure coding in windows azure storage,'' in \emph{Presented
  as part of the 2012 USENIX Annual Technical Conference (USENIX ATC 12)},
  2012, pp. 15--26.

\bibitem{muralidhar2014f4}
S.~Muralidhar, W.~Lloyd, S.~Roy, C.~Hill, E.~Lin, W.~Liu, S.~Pan, S.~Shankar,
  V.~Sivakumar, L.~Tang \emph{et~al.}, ``f4: Facebook's warm blob storage
  system,'' in \emph{11th USENIX Symposium on Operating Systems Design and
  Implementation (OSDI 14)}, 2014, pp. 383--398.

\bibitem{ahmad2014shufflewatcher}
F.~Ahmad, S.~T. Chakradhar, A.~Raghunathan, and T.~Vijaykumar,
  ``Shufflewatcher: Shuffle-aware scheduling in multi-tenant mapreduce
  clusters,'' in \emph{2014 USENIX Annual Technical Conference (USENIX ATC
  14)}, 2014, pp. 1--13.

\bibitem{ernvall2013capacity}
T.~Ernvall, S.~El~Rouayheb, C.~Hollanti, and H.~V. Poor, ``Capacity and
  security of heterogeneous distributed storage systems,'' \emph{IEEE Journal
  on Selected Areas in Communications}, vol.~31, no.~12, pp. 2701--2709, 2013.

\bibitem{akhlaghi2010fundamental}
S.~Akhlaghi, A.~Kiani, and M.~R. Ghanavati, ``A fundamental trade-off between
  the download cost and repair bandwidth in distributed storage systems,'' in
  \emph{2010 IEEE International Symposium on Network Coding (NetCod)}.\hskip
  1em plus 0.5em minus 0.4em\relax IEEE, 2010, pp. 1--6.

\bibitem{gaston2013realistic}
B.~Gast{\'o}n, J.~Pujol, and M.~Villanueva, ``A realistic distributed storage
  system that minimizes data storage and repair bandwidth,'' \emph{arXiv
  preprint arXiv:1301.1549}, 2013.

\bibitem{tebbi2014code}
M.~A. Tebbi, T.~H. Chan, and C.~W. Sung, ``A code design framework for
  multi-rack distributed storage,'' in \emph{Information Theory Workshop (ITW),
  2014 IEEE}.\hskip 1em plus 0.5em minus 0.4em\relax IEEE, 2014, pp. 55--59.

\bibitem{7541298}
Y.~Hu, P.~P.~C. Lee, and X.~Zhang, ``Double regenerating codes for hierarchical
  data centers,'' in \emph{2016 IEEE International Symposium on Information
  Theory (ISIT)}, July 2016, pp. 245--249.

\bibitem{dimakis2011survey}
A.~G. Dimakis, K.~Ramchandran, Y.~Wu, and C.~Suh, ``A survey on network codes
  for distributed storage,'' \emph{Proceedings of the IEEE}, vol.~99, no.~3,
  pp. 476--489, 2011.

\bibitem{bang2008digraphs}
J.~Bang-Jensen and G.~Z. Gutin, \emph{Digraphs: theory, algorithms and
  applications}.\hskip 1em plus 0.5em minus 0.4em\relax Springer Science \&
  Business Media, 2008.

\end{thebibliography}

\end{document}